\documentclass[prl,superscriptaddress, twocolumn]{revtex4}

\usepackage{amsfonts,amssymb,amsmath}
\usepackage[]{graphics,graphicx,epsfig,wrapfig}
\usepackage{amsthm}
\usepackage{enumitem}
\usepackage{color}
\usepackage[noxcolor]{pstricks} 
\usepackage{pst-all} 
\usepackage[color]{changebar}
\usepackage{booktabs}
\usepackage{multirow}
\usepackage{siunitx}

\def\identity{\leavevmode\hbox{\small1\kern-3.8pt\normalsize1}}

\renewcommand{\epsilon}{\varepsilon}

\newtheorem{definition}{Definition}
\newtheorem{prop}[definition]{Proposition}

\newtheorem*{rep@theorem}{\rep@title}
\newcommand{\newreptheorem}[2]{%
\newenvironment{rep#1}[1]{%
 \def\rep@title{#2 \ref{##1} (restatement)}%
 \begin{rep@theorem}}%
 {\end{rep@theorem}}}
\makeatother

\newreptheorem{thm}{Theorem}
\newreptheorem{lem}{Lemma}

\def\ba#1\ea{\begin{align}#1\end{align}}
\def\ban#1\ean{\begin{align*}#1\end{align*}}

\newcommand{\be}{\begin{equation}}
\newcommand{\ee}{\end{equation}}

\def\benum{\begin{enumerate}}
\def\eenum{\end{enumerate}}

\def\squareforqed{\hbox{\rlap{$\sqcap$}$\sqcup$}}
\def\qed{\ifmmode\squareforqed\else{\unskip\nobreak\hfil
\penalty50\hskip1em\null\nobreak\hfil\squareforqed
\parfillskip=0pt\finalhyphendemerits=0\endgraf}\fi}
\def\endenv{\ifmmode\;\else{\unskip\nobreak\hfil
\penalty50\hskip1em\null\nobreak\hfil\;
\parfillskip=0pt\finalhyphendemerits=0\endgraf}\fi}
\newcommand{\tr}{\text{tr}}

\newcommand{\<}{\langle}
\renewcommand{\>}{\rangle}

\def\be{\begin{equation}}
\def\ee{\end{equation}}
\def\ben{\begin{eqnarray}}
\def\een{\end{eqnarray}}

\def\bei{\begin{itemize}}
\def\eei{\end{itemize}}

\mathchardef\ordinarycolon\mathcode`\:
\mathcode`\:=\string"8000
\def\vcentcolon{\mathrel{\mathop\ordinarycolon}}
\begingroup \catcode`\:=\active
  \lowercase{\endgroup
  \let :\vcentcolon
  }

\newcommand{\nc}{\newcommand}
 \nc{\proj}[1]{|#1\rangle\!\langle #1 |} 
\nc{\avg}[1]{\langle#1\rangle}

\nc{\conv}{\operatorname{conv}}
\nc{\ox}{\otimes} \nc{\dg}{\dagger} \nc{\dn}{\downarrow}
\nc{\lmax}{\lambda_{\text{max}}}
\nc{\lmin}{\lambda_{\text{min}}}

\nc{\csupp}{{\operatorname{csupp}}}
\nc{\qsupp}{{\operatorname{qsupp}}} \nc{\var}{\operatorname{var}}
\nc{\rar}{\rightarrow} \nc{\lrar}{\longrightarrow}
\nc{\poly}{\operatorname{poly}}
\nc{\polylog}{\operatorname{polylog}} \nc{\Lip}{\operatorname{Lip}}
\nc{\Om}{\Omega}
\nc{\wt}[1]{\widetilde{#1}}

\def\>{\rangle}
\def\<{\langle}

\nc{\glneq}{{\raisebox{0.6ex}{$>$}  \hspace*{-1.8ex} \raisebox{-0.6ex}{$<$}}}
\nc{\gleq}{{\raisebox{0.6ex}{$\geq$}\hspace*{-1.8ex} \raisebox{-0.6ex}{$\leq$}}}

\nc{\vholder}[1]{\rule{0pt}{#1}}
\nc{\wh}[1]{\widehat{#1}}
\nc{\h}[1]{\widehat{#1}}

\nc{\ob}[1]{#1}

\def\beq{\begin {equation}}
\def\eeq{\end {equation}}

\def\be{\begin{equation}}
\def\ee{\end{equation}}

\nc{\eq}[1]{(\ref{eq:#1})} 
\nc{\eqs}[2]{\eq{#1} and \eq{#2}}

\nc{\eqn}[1]{Eq.~(\ref{eqn:#1})}
\nc{\eqns}[2]{Eqs.~(\ref{eqn:#1}) and (\ref{eqn:#2})}

\nc{\region}{\cS\cW}
\newenvironment{protocol*}[1]
  {
    \begin{center}
      \hrulefill\\
      \textbf{#1}
  }
  {
    \vspace{-1\baselineskip}
    \hrulefill
    \end{center}
  }

\begin{document}

\title{Activation of monogamy in non-locality using local contextuality }
\author{Debashis \surname{Saha}}
%\email{debashis.saha@phdstud.ug.edu.pl}
\affiliation{Institute of Theoretical Physics and Astrophysics, National Quantum Information Center, Faculty of Mathematics, Physics and Informatics, University of Gda\'{n}sk, 80-952 Gda\'{n}sk, Poland}
\author{Ravishankar \surname{Ramanathan}}
%\email{ravishankar.r.10@gmail.com}
\affiliation{Institute of Theoretical Physics and Astrophysics, National Quantum Information Center, Faculty of Mathematics, Physics and Informatics, University of Gda\'{n}sk, 80-952 Gda\'{n}sk, Poland}

\begin{abstract}
A unified view on the phenomenon of monogamy exhibited by Bell inequalities and non-contextuality inequalities arising from the no-signaling and no-disturbance principles is presented using the graph-theoretic method introduced in \textit{Phys. Rev. Lett. 109, 050404 (2012)}. 
We propose a novel type of trade-off, namely Bell inequalities that do not exhibit monogamy features of their own can be activated to be monogamous by the addition of a local contextuality term. This is illustrated by means of the well-known $\mathcal{I}_{3322}$ inequality, and reveals a resource trade-off between bipartite correlations and the local purity of a single system.
%The non-necessity of the condition for the existence of a monogamy relation is shown by means of an explicit counter-example. 
In the derivation of novel no-signaling monogamies, we uncover a new feature, namely that two-party Bell expressions that are trivially classically saturated can become non-trivial upon the addition of an expression involving a third party with a single measurement input. 
\end{abstract}

\maketitle

{\it Introduction.} The Bell theorem \cite{Bell} and the Kochen-Specker theorem \cite{KS} are cornerstone results in foundations of quantum mechanics. 
While the Bell theorem demonstrates that local hidden variable theories are incompatible with the statistical predictions of quantum mechanics, the Kochen-Specker (KS) theorem demonstrates the incompatibility of quantum theory with the assumption of ``outcome non-contextuality" when describing systems with more than two distinguishable states. In other words, the KS theorem shows that there are quantum measurements whose outcomes cannot be predefined in a noncontextual manner, i.e., independent of the measurement's context (the choice of jointly measurable tests that may be performed together). 
However, Bell theorem can be interpreted as a specific case of this feature where the measurement's context is remote. 

The non-local correlations between spatially separated systems that lead to the violation of the well-known Clauser-Horne-Shimony-Holt (CHSH) inequality exhibit the phenomenon of \textit{monogamy} \cite{Toner}. In particular, it was proved by Toner in \cite{Toner} that in a Bell experiment with three spatially separated parties Alice, Bob and Charlie, when Alice and Bob observe a violation of the CHSH inequality between their systems, the no-signaling principle imposes that Alice and Charlie cannot at the same time observe a violation of the inequality. This feature of non-local correlations reflects the monogamy of entanglement in the underlying quantum state used in the Bell experiment, and is significant in cryptographic scenarios \cite{Masanes, Pawlowski} where Alice and Bob can verify a sufficient Bell violation to guarantee that their systems are not much correlated with any eavesdropper's system. In fact, such a trade-off in the non-local correlations can be seen for every Bell inequality \cite{PB}. 
%To be precise, consider any inequality $\mathcal{I}_{AB} \leq \mathcal{L}_{I}$ in the $(2,m_A,m_B, k_A, k_B)$ Bell scenario, i.e., two parties with $m_A$ inputs for Alice and $m_B$ inputs for Bob with $k_A$ and $k_B$ outputs respectively. When this Bell experiment is performed by Alice simultaneously with $m_B$ Bobs, a monogamy relation manifests itself, i.e. there holds
%%for an inequality with $m_A$ settings for Alice and $m_B$ settings for Bob, a trade-off in the violation due to the no-signaling constraint is seen when the Bell experiment is performed by Alice together with $m$ Bobs. This non-locality monogamy relation for an inequality of the form $\mathcal{I}_{AB} \leq R$ takes the form
%\begin{equation}
%\sum_{i=1}^{m_B} \mathcal{I}_{A B^{(i)}} \stackrel{NS}{\leq} m_B \mathcal{L}_I. 
%\end{equation}
Other (stricter) no-signaling monogamy relations for specific classes of Bell inequalities have also been discovered \cite{TDS, AGCA, RH, ADPTA}. 

A general constraint analogous to no-signaling is also imposed in a single system contextuality scenario. This is the \textit{no-disturbance principle} \cite{RSK} which is a consistency constraint that the probability distribution observed for the outcomes of measurement of any observable is independent of which set of other (co-measurable) observables it is measured alongside. It was shown in \cite{RSK} that this constraint also imposes a trade-off in the violations of non-contextuality inequalities on a single system, the simplest of these inequalities being the Klyachko-Can-Binicoglu-Shumovsky (KCBS) inequality \cite{KCBS}. Recently, the phenomenon of monogamy between the non-local correlations and the single system contextuality has been pointed out by Kurzynski-Cabello-Kaszlikowski \cite{CKK14} and has been experimentally verified \cite{expt}. In particular, \textit{any} violation of the CHSH inequality between Alice and Bob's systems implies that locally Alice's system alone does not exhibit a violation of the KCBS inequality and vice versa. 
%, i.e., 
%\begin{equation}
%\mathcal{I}^{\tiny{chsh}}_{AB} + \mathcal{I}^{\tiny{kcbs}}_{A} \stackrel{ND}{\leq} \mathcal{L}_{\tiny{chsh}} + \mathcal{L}_{\tiny{kcbs}}.
%\end{equation}

In this letter, we first outline a sufficient condition to derive no-signaling and no-disturbance monogamies using a graph-theoretic method, while 
%finding appropriate \textit{induced} chordal decompositions of the \textit{commutation graph} corresponds to the set of observables. Although an appropriate chordal decomposition is sufficient to derive a monogamy relation, we 
showing that this condition is not also necessary by means of a counter-example. 
%While the no-signaling value of a Bell inequality and the no-disturbance value of a non-contextuality inequality are both calculable by means of a linear program, the graph-theoretic 
This method allows us to derive new types of monogamies not previously studied in the literature. For instance, we show any two cycle inequalities for any lengths of the cycle, whether studied as non-contextuality \cite{LSW,AQBCC} or Bell inequalities \cite{BC}, exhibit a monogamy relation in both contextual and nonlocal scenarios under certain conditions. 
Then we proceed to our main result that a non-monogamous Bell inequality can be activated to be monogamous by the addition of a local (state-dependent) non-contextuality inequality. This novel type of monogamy, which we illustrate via the well-known $\mathcal{I}_{3322}$ Bell inequality reveals a resource trade-off between bipartite correlations and the local purity of a single system. This is a theory-independent non-locality analogue of the well-known Coffman-Kundu-Wootters trade-off \cite{CKW} for entanglement. 
%The fact that state-independent contextuality cannot replace the local term lends strength to the resource trade-off revealed by these monogamies. 
 Finally, in investigating methods for the derivation of novel no-signaling monogamies, we uncover a new feature, namely that Bell expressions that are trivially satisfied by a classical Alice and Bob, can give rise to no-signaling violations upon the addition of a third party with a single measurement input. We explore this counter-intuitive feature by means of an explicit example.   \\

{\it The graph-theoretic method to derive monogamy relations.}
Let us first state and explain the graph-theoretic method to derive general monogamy relations for both non-contextuality and Bell inequalities \cite{RSK}.
% To do this, we will first describe the notion of a \textit{commutation graph} that is used to represent the observables measured in the setup. 
%The set of all observables that appear in the combined Bell and contextuality experiment is represented by means of a commutation graph, where the vertices of the commutation graph represent the observables that appear in the Bell and non-contextuality inequalities and edges connect commuting observables. 
%In what follows, we will always consider the commutation graph to be finite, simple and undirected. 
All the observables $\{A_1, \dots, A_n\}$ that appear in the combined Bell and contextuality experiment, are represented by means of a commutation graph $G$, that is constructed as follows. Each vertex of $G$ represents an observable and two vertices are connected by an edge if the corresponding observables can be measured together.
Notice that in the commutation graph each clique represents a context, i.e., a jointly measurable system of observables.  \\
A graph is said to be chordal if all cycles of length four or more in the graph have a chord, i.e., an edge connecting two non-adjacent vertices in the cycle. As explained in the Supplemental Material, chordal graphs are known to have equivalent characterizations in terms of admitting a maximal clique tree as well as having a simplicial elimination scheme \cite{HL09}. The importance of chordal commutation graphs comes from the following proposition \cite{RSK} stating that a set of observables admits a joint probability distribution for its outcomes when its corresponding commutation graph is chordal. This statement can be seen as the generalization of the statement by Fine \cite{Fine82} showing that tree graphs admit a joint probability distribution. A similar statement in the context of relational databases was proven in \cite{BFMY83}. 
%\begin{prop}
%\label{prop-chordal}[see \cite{RSK}]
%A set of observables $A_1, A_2, \dots, A_n$ admits a joint probability distribution if the commutation graph $G(A_1, \dots, A_n)$ is chordal. 
%\end{prop}
For completeness, we give an explicit proof in the Supplementary Material \cite{SM} using the notion of simplicial elimination ordering for chordal graphs. It is also noteworthy that the above condition is not necessary for existence of a monogamy relation as discussed in the Supplementary Material \cite{SM}. 

\textit{Method}. 
{Consider a set of Bell and non-contextuality inequalities $\{\mathcal{I}_k \}$, and let $\omega_c = \sum_k \omega_c(\mathcal{I}_k)$ denote the maximum classical value of the combined expression, i.e., $\sum_k \mathcal{I}_k \leq \omega_c$ in any classical theory. Let $G_{\{\mathcal{I}_k\}}$ denote the commutation graph representing the observables measured by all the parties in the non-contextuality or Bell scenario. A no-signaling or no-disturbance monogamy relation holds if $G_{\{\mathcal{I}_k\}}$ can be decomposed into a set of induced chordal subgraphs $\{G_{\text{c-sub}}^{(j)}\}$, such that the sum of the algebraic values of the reduced Bell expressions in each of the chordal graphs equals $\omega_c$. }
%The key idea in deriving monogamy relations is therefore to decompose the commutation graph representing \textit{all} the observables in the setup into chordal graphs in such a way that the sum of the algebraic values of the expressions in each of the chordal graphs equals the sum of the classical values of the Bell inequalities in the setup. 
% as shown in \cite{PB,symext}, 
%\begin{proof}

{\bf Lemma.}
\textit{It is sufficient to consider \textit{induced} chordal subgraphs in the method described above. }\\
%That is, if any chordal decomposition exists to reproduce the sum of the classical values $\omega_c$, then a decomposition into induced subgraphs also necessarily exists to do so. 
To see this, suppose that a set of chordal graphs $G_{\text{c-sub}}^{(j)}$ with associated (classical equals no-signaling) values for the reduced Bell expressions $\omega_c^{(j)}$ exists satisfying $\sum_j \omega_c^{(j)} = \omega_c$. Consider the optimal classical deterministic strategy achieving value $\omega_c$ for the whole graph, this strategy must therefore necessarily achieve value $\omega_c^{(j)}$ on each of the chordal subgraphs. In other words, a set of optimal and compatible classical (deterministic) strategies for all of the $G_{\text{c-sub}}^{(j)}$ exists, where compatibility denotes that the values are assigned by the strategies to any observable $A_m$ is the same in each of the chordal subgraphs $A_m$ appears in. From this observation, it follows that each $G_{\text{c-sub}}^{(j)}$ may be taken to be induced, i.e., any edges between two vertices $A_i, A_j \in V(G_{\text{c-sub}}^{(j)})$ that are present in $G_{\{\mathcal{I}_k\}}$ may also be included in $G_{\text{c-sub}}^{(j)}$. 

However, remark that when considering different Bell expressions, it may happen that the classical deterministic strategies for these expressions are not compatible with each other, in the sense that the same observable $A_i$ is assigned different values in the optimal strategies for different expressions. In such cases one might have classical ``monogamies" where the classical value of the sum is strictly smaller than the sum of the individual classical values, i.e., $\omega_c < \sum_k \omega_c(\mathcal{I}_k)$ \cite{RSK, RH}.

{\it Cycle inequalities.-} 
As mentioned earlier, any Bell inequality can also be viewed as a non-contextuality inequality 
%(in terms of correlation of sequential measurements) 
on the combined system of distant parties. By incorporating other Bobs into Alice's system monogamy relations in nonlocal-contextual and contextual-contextual scenarios can be inferred from those in the Bell scenario. 
%For instance, one can interpret the usual monogamy relation for the CHSH inequality instead as monogamy in the nonlocal-contextual scenario in which Alice possesses higher than three dimensional system and two non-commuting observables measured on her system (out of the four observables that she measures) appear in a nonlocal CHSH inequality with Bob. 
But the most interesting case is when the contextuality test is performed on a single system (for example qutrit) which does not exhibit nonlocality. Recently, $n$-cycle ($n\geq 5$) non-contextual inequalities have been proposed and shown to be maximally violated by qutrits \cite{LSW,AQBCC}. The analogous cycle Bell inequality \cite{BC} with $n$ inputs each is also studied. Motivated by these facts, we first consider cycle inequalities that are the simplest nontrivial case to study unified monogamy relations. It is shown that monogamy exists for two $n$-cycle Bell inequality of same length in nonlocal-nonlocal scenario \cite{BKP}. In the nonlocal-contextual scenario, monogamy between CHSH and $n$-cycle non-contextuality inequality, and in the contextual-contextual scenario monogamy between any two cycle inequalities are pointed out \cite{JWG}. Here, we show a more general result in this direction. \\
{\bf Proposition.}
\textit{ Any two cycle inequalities with different length of the cycle, having  at least two common observables, are monogamous 
in any theory satisfying the no-disturbance principle. For the monogamy to hold in the nonlocal-contextual and contextual-contextual scenario, suitable additional commutation relations are required.}

The proof of this Proposition with the decompositions of induced chordal subgraphs is explicitly described in Supplementary Material \cite{SM}. This result can also be generalized for many outcome cycle inequalities \cite{BKP} in all three scenarios (see Supplementary Material \cite{SM}).\\

{\it Activation of monogamy relation in Bell inequality.-}
The monogamy relations for entanglement establish a strict trade-off in the shareability of this resource. In \cite{CKW}, Coffman, Kundu and Wootters established a monogamy relation for the tangle, which is a well-known measure of entanglement.
% defined as follows. For a two-qubit state $\rho_{AB}$, let $\lambda_i$ for $i = 1,2,3,4$ denote the square roots of the eigenvalues in decreasing order of the operator $\rho_{AB} (\sigma_y \otimes \sigma_y) \rho^*_{AB} (\sigma_y \otimes \sigma_y)$. The tangle $\tau_{AB}$ is then defined as $\tau_{AB} = \left[ \max\{ \lambda_1 - \lambda_2 - \lambda_3 - \lambda_4 , 0 \} \right]^2$. 
The monogamy relation for the tangle reads as:
\begin{equation}\label{CKW}
\tau^{(1)}_{A|BC} \geq \tau^{(2)}_{AB} + \tau^{(2)}_{AC}, 
\end{equation}
where $\tau^{(1)}_{A|BC}$ is the tangle between qubit $A$ and the pair $BC$. 
This relation says that amount of entanglement that qubit $A$ has with $BC$, cannot be less than the sum of the individual entanglements with qubits $B$ and $C$ separately. 

We now propose an analogous relationship between non-locality and contextuality. 
%In this tripartite scenario, Alice, Bob and Charlie share qutrit-qubit-qubit state. 
Consider the nonlocal-contextual scenario with three observers Alice, Bob and Charlie. Alice performs six possible $\pm1-$valued measurements, say, $A_i$, with $i \in \{1,...,6\}$. In each run of the experiment, Alice randomly measures two compatible observables from the set $\{A_1A_{2},A_2A_{3},A_3A_{4},A_4A_{5},A_5A_{1}\}$, or a single observable $A_6$ on her sub-system. While Bob and Charlie randomly performs one of their three $\pm1-$valued observables, say $B_1,B_2,B_3$ and $C_1,C_2,C_3$, on their respective subsystems. The $\mathcal{I}_{3322}$ inequality \cite{CG04} involving Alice's observables $A_1,A_4,A_6$ and Bob's observable $B_1,B_2,B_3$ is given by,
\begin{equation}\label{i3322}
\begin{split}
\mathcal{I}^{A_1A_4A_6}_{B_1B_2B_3} = & \langle A_1 \rangle + \langle A_4 \rangle +\langle B_1 \rangle + \langle B_2 \rangle - \langle A_1B_1 \rangle \\
& - \langle A_1B_2 \rangle - \langle A_1B_3 \rangle - \langle A_4B_1 \rangle - \langle A_4B_2 \rangle  \\
& + \langle A_4B_3 \rangle - \langle A_6B_1 \rangle + \langle A_6B_2 \rangle \leq 4.
\end{split}\end{equation} 
A no-signaling box exists that achieves the value $8$ for this inequality. This is a tight Bell inequality in the scenario of two parties, with three dichotomic inputs each, which exhibits several remarkable properties \cite{CG04,PV10,VW11}. 
%Firstly, it was shown in \cite{CG04} that the $\mathcal{I}_{3322}$ inequality is inequivalent to the CHSH inequality that there are two qubit states that violate $\mathcal{I}_{3322}$ that do not violate the latter. Secondly, numerical investigations on the maximum quantum violation of this inequality have suggested that this maximum value is only reached in the limit of an infinite dimensional Hilbert space \cite{PV10}. Thirdly, it was shown in \cite{VW11} that for all dimensions $d \geq 0$ and any observables, use of the maximally entangled state in that dimension does not lead to the maximum quantum violation. 
With regards to its monogamy properties, it was shown in \cite{CG04} that the non-locality that it reveals can be shared, i.e., there exist three qubit states $| \psi \rangle_{ABC}$ such that both $\rho_{AB} = \tr_C(|\psi \rangle \langle \psi|)$ and $\rho_{AC} = \tr_B(|\psi \rangle \langle \psi |)$ violate the inequality at the same time. Here, we show another remarkable property of the $\mathcal{I}_{3322}$ inequality, namely that the addition of a minimal local state-dependent contextuality term can make the $\mathcal{I}_{3322}$ inequality monogamous. More precisely, 
%in the precise sense outlined below.     
%We want to establish a monogamy relation between the bipartite correlations Alice-Bob, the correlations Alice-Charlie, and the correlations between commuting measurements on Alice's subsystem as,
\begin{equation}\label{i3322mr}
\mathcal{I}^{A_1A_4A_6}_{B_1B_2B_3}  + \mathcal{I}^{A_1A_4A_6}_{C_1C_2C_3} + 2 \mathcal{I}(5) \overset{ND}{\leq} 14,
\end{equation}
where $\mathcal{I}(5)=\sum\limits_{i=1}^{4} \langle A_iA_{i+1}\rangle -\langle A_5A_{1}\rangle \leq 3$ is the non-contextual bound.
\begin{proof} To show the validity of the above inequality, lets decompose the commutation graph of the whole quantity in the following way,
\begin{equation}\label{cd:i3322}
\begin{split}
& G\left[\mathcal{I}^{A_1A_4A_6}_{B_1B_2B_3}\right]  + G\left[\mathcal{I}^{A_1A_4A_6}_{C_1C_2C_3}\right] + 2 \times G\left[\mathcal{I}(5)\right] \rightarrow \\
& G\left[\mathcal{I}^{A_1A_4A_6}_{B_1,C_2,A_5} \right] + G\left[\mathcal{I}^{A_1A_4A_6}_{B_2,C_1,A_5}\right] + G\left[\mathcal{I}(5)^{A_1A_2A_3A_4}_{B_3}\right] \\
& + G\left[\mathcal{I}(5)^{A_1A_2A_3A_4}_{C_3}\right] 
\end{split}\end{equation} 
where $G[\mathcal{I}]$ denotes the commutation graph corresponding to the set of observables appear in the inequality $\mathcal{I}$. It can be checked (as shown in Fig. \ref{I3322}) that all the decomposed induced graphs are chordal and hence possess joint probability distributions.
\end{proof}

\begin{figure}[!htb]
\centering
\includegraphics[scale=0.61]{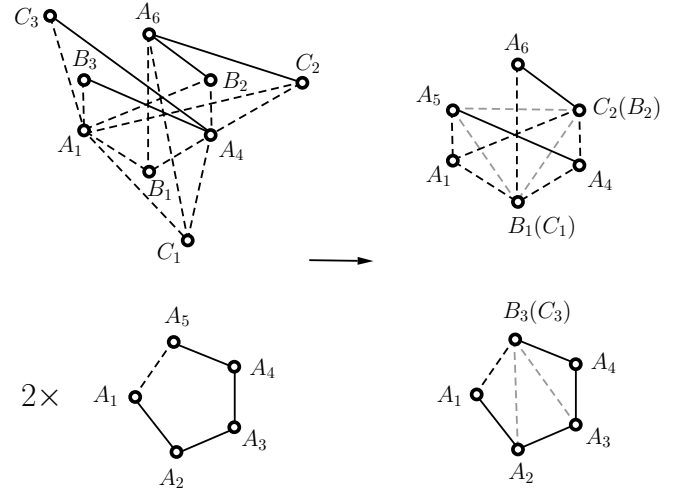}
\caption{Chordal decomposition given in Eq. \eqref{cd:i3322} to show the monogamy relation  between the bipartite $\mathcal{I}_{3322}$ correlations of Alice-Bob, Alice-Charlie, and the correlations between commutating measurements on Alice's subsystem. Here the \textit{dashed} line represents the contradiction edge and \textit{grey} lines denote the additional commutation relations.}
\label{I3322}
\end{figure}

To see the correspondence between \eqref{CKW} and \eqref{i3322mr}, one can reinterpret the monogamy relation \eqref{CKW} in terms of the purity of the subsystem $A$ (denoted by $\mathfrak{p}_A$) as, $\tau_{AB}+\tau_{AC}+\mathfrak{p}_A \leq 1$. Similarly, the purity of a system can be related to the resource of state dependent contextuality \cite{RH1}. 
This generalizes the view presented in \cite{CKK14} where a single CHSH inequality was shown to have a trade-off in violation with a local contextuality term. \\
In the above scenario on the other hand, it is worth noting that a single $\mathcal{I}_{3322}$ inequality does not exhibit a monogamy with the local non-contextuality inequality $\mathcal{I}(5)$. In fact, a no-signaling box that violates both inequalities can be found \cite{SM}.

\textit{Lack of monogamies due to non-trivial Bell inequalities with single inputs for some parties.-}
In this section, we study a novel feature of Bell inequalities that appears in the derivation of no-signaling monogamy relations. There exists two-party Bell expression which has the same classical and no-signaling value, can be turned into non-trivial inequality upon the addition of an expression involving a third party with a single measurement input.  We present such an example in the scenario of three parties with $3,2$ and $1$ inputs respectively and four outputs per setting. While this appears counter-intuitive at first sight, we explain that this arises due to an incompatibility between the optimal classical strategies for the sub-expressions into which the commutation graph of the whole Bell expression is decomposed.\\
In \cite{RH}, it was shown that in the paradigmatic example of correlation inequalities for binary outcomes, the parameter known as the \textit{contradiction number} gives a sufficient characterization of the monogamy. Namely, that if the removal of a certain number $m$ observables of any one party ($m$ is called the \textit{contradiction number} of the inequality) results in the residual expression having a local hidden variable description, then a monogamy manifests itself when Alice performs the correlation Bell experiment with $m+1$ Bobs. The existence of non-trivial Bell inequalities with single inputs for some of the parties as stated above then implies that this result does not readily extend to general inequalities for many outcomes. 
%and this arises due to an incompatibility between the optimal classical strategies for the sub-expressions into which the commutation graph of the whole Bell expression is decomposed. 
%As explained in the text, such incompatibilities do not arise in the decomposition into induced chordal graphs 

%Namely, there exist correlation inequalities for more outcomes with contradiction number one, 
%This arises because even though thi
%due to an incompatibility in the optimal classical strategies for the sub-expression  
%Here, we consider a natural generalization of the correlation inequalities into more outcomes, to the so-called linear games \cite{}.

We study a specific example of an inequality $\mathcal{I}^{A_1A_2A_3}_{B_1B_2B_3}$ which has a \textit{contradiction number} of one, i.e., upon the removal of the input $B_3$ for Bob, the residual expression $\mathcal{I}^{A_1A_2A_3}_{B_1B_2}$ admits a local hidden variable description (a joint probability distribution $P(A_1, A_2, A_3, B_1, B_2)$). The residual expression $\mathcal{I}^{A_1A_2A_3}_{C_1}$ also evidently admits a classical description $P(A_1,A_2,A_3,C_1)$ simply because Charlie only measures a single input in this expression. Intuitively, one might expect that the expression $\mathcal{I}^{A_1A_2A_3}_{B_1B_2C_1}$ also admits such a description, for instance via the Fine trick \cite{Fine82}
\begin{equation}
\begin{split}
& P(A_1,A_2,A_3,B_1,B_2,C_1) \\
& = \frac{P(A_1,A_2,A_3,B_1,B_2)P(A_1,A_2,A_3,C_1)}{P(A_1,A_2,A_3)}.
\end{split}\end{equation}
Here the subtlety arises. The optimal strategy that achieves the classical value for $\mathcal{I}^{A_1A_2A_3}_{B_1B_2}$ need not give rise to the same marginal distribution $P(A_1,A_2,A_3)$ as the optimal strategy for $\mathcal{I}^{A_1A_2A_3}_{C_1}$. This implies that the above construction does not automatically work, and therefore one might encounter violation of Bell expressions $\mathcal{I}^{A_1A_2A_3}_{B_1B_2C_1}$ for which $\mathcal{I}^{A_1A_2A_3}_{B_1B_2}$ is trivially saturated by a classical strategy and yet the contradiction arises from a third party measuring a single setting. 
%In other words, the addition of a trivial expression for a third party who only measures a single setting results in the set of observables no longer admitting a joint probability distribution. 
We give below an example of a Bell expression having such a property. \\
This is a Bell expression $\mathcal{I}^{A_1A_2A_3}_{B_1B_2B_3}$ belonging to the family of tight Bell expressions in the $(2,3,3,4,4)$ scenario found by Cabello \cite{Cab12}, and has been experimentally tested \cite{Cabexpt12}. The expression of interest to us is actually the reduced Bell expression $\mathcal{I}^{A_1A_2A_3}_{B_1B_2C_1}$ given as
\begin{eqnarray}\label{eq:cab12bi}
\mathcal{I}^{A_1A_2A_3}_{B_1B_2C_1} & = \langle A_1^{10} B_1^{10} \rangle + \langle A_1^{01} B_2^{10} \rangle + \langle A_1^{11} C_1^{10} \rangle + \nonumber \\
&\langle A_2^{10} B_1^{01} \rangle + \langle A_2^{01} B_2^{01} \rangle + \langle A_2^{11} C_1^{01} \rangle +  \\
&\langle A_3^{10} B_1^{11} \rangle + \langle A_3^{01} B_2^{11} \rangle - \langle A_3^{11} C_1^{11} \rangle \leq 7, \nonumber 
\end{eqnarray}
where $\langle A_k^{i_k,j_k} B_l^{i_l,j_l} \rangle$ denotes the mean value of the product of the $i_k$-th and $j_k$-th bit (assigned $\pm 1$ values) of the result of measuring $A_k$ times the $i_l$-th and $j_l$-th bit of the result of measuring $B_l$ (respectively $C_l$). Evidently a classical strategy can be found for the residual expression $\mathcal{I}^{A_1A_2A_3}_{B_1B_2}$ when $B_3$ is removed, for instance Alice and Bob output $00$. The residual expression $\mathcal{I}^{A_1A_2A_3}_{C_1}$ also trivially admits a classical strategy, for instance Charlie outputs $00$ and Alice outputs $00$ for inputs $A_1$ and $A_2$ and $10$ for input $A_3$. On the other hand, no classical strategy exists to saturate both expressions simultaneously, and it can be shown that $\mathcal{I}^{A_1A_2A_3}_{B_1B_2C_1} \leq 7$ in any local hidden variable theory. From the extremal box of the no-signaling polytope given in table \ref{nsbox}, it can be verified that 
the algebraic value of the inequality is 9. This implies that the Bell inequalities and the corresponding classical polytopes with single inputs for some of the parties have some interesting property.

\begin{table}[http]
\begin{center}
\caption{In the following table the no-signaling probability distribution for which the tripartite inequality Eq. \eqref{eq:cab12bi} attains its algebraic value 9, is described. Each column corresponds to six different inputs $A_k, B_l, C_1$, and for a particular input only the events $P(a_0^{(k)}a_{1}^{(k)}, b_0^{(l)}b_1^{(l)}, c_0^{(1)}c_1^{(1)} | A_k, B_l, C_1)$  with non-zero probability are listed where the probability of each of these events is  $\frac{1}{8}$.} \label{nsbox}
  \begin{tabular}{ | m{1.3cm} | m{1.3cm} |  m{1.3cm} |  m{1.3cm} | m{1.3cm} | m{1.3cm} | }
  \hline
    \multicolumn{1}{|c|}{$A_1 B_1 C_1$} &\multicolumn{1}{c}{$A_1 B_2 C_1$} &\multicolumn{1}{|c|}{$A_2 B_1 C_1$} &\multicolumn{1}{|c|}{$A_2 B_2 C_1$} &\multicolumn{1}{|c|}{$A_3 B_1 C_1$} &\multicolumn{1}{|c|}{$A_3 B_2 C_1$}\\
    \hline
        \scriptsize{(00, 00, 00)}  &  \scriptsize{(00, 00, 00)}  &\scriptsize{(00, 00, 00)} & \scriptsize{(00, 00, 00)}& \scriptsize{(00, 00, 01)} & \scriptsize{(00, 00, 01)} \\ 
   \scriptsize{(00, 00, 01)}  &  \scriptsize{(00, 00, 01)}  & \scriptsize{(00, 00, 10)} &\scriptsize{(00, 10, 10)}  & \scriptsize{(00, 00, 10)}& \scriptsize{(00, 11, 10)} \\ 
         \scriptsize{(00, 01, 00)}    & \scriptsize{(00, 01, 00)}  & \scriptsize{(01, 00, 01)} &\scriptsize{(01, 01, 01)} &\scriptsize{(01, 00, 00)} &  \scriptsize{(01, 01, 00)}\\ 
          \scriptsize{(00, 01, 01)} &  \scriptsize{(00, 01, 01)}  & \scriptsize{(01, 00, 11)} & \scriptsize{(01, 11, 11)}& \scriptsize{(01, 00, 11)} & \scriptsize{(01, 10, 11)}\\ 
		    \scriptsize{(01, 00, 10)} &  \scriptsize{(01, 10, 10)}   &\scriptsize{(10, 01, 01)} & \scriptsize{(10, 00, 01)} & \scriptsize{(10, 01, 00)} & \scriptsize{(10, 00, 00)}\\ 
        \scriptsize{(01, 00, 11)}   &\scriptsize{(01, 10, 11)} & \scriptsize{(10, 01, 11)}& \scriptsize{(10, 10, 11)} &\scriptsize{(10, 01, 11)} & \scriptsize{(10, 11, 11)} \\ 
         \scriptsize{(01, 01, 10)}  &\scriptsize{(01, 11, 10)}& \scriptsize{(11, 01, 00)}& \scriptsize{(11, 01, 00)}& \scriptsize{(11, 01, 01)}& \scriptsize{(11, 01, 01)} \\ 
          \scriptsize{(01, 01, 11)} & \scriptsize{(01, 11, 11)}& \scriptsize{(11, 01, 10)} & \scriptsize{(11, 11, 10)}& \scriptsize{(11, 01, 10)}& \scriptsize{(11, 10, 10)}  \\ \hline
  \end{tabular}
  
\end{center}
\end{table}

%{\it Discussion.}

{\it Conclusions.} In this paper, we have presented a unified view on monogamy relations for Bell and non-contextuality inequalities derived from the physical principles of no-signaling and no-disturbance. The unification was achieved by considering a graph-theoretic decomposition of the graph representing all the observables in the experiment into induced chordal subgraphs. We have used this method to show that any two generalized cycle inequalities (for any number of outputs) exhibit a monogamy relation. As a main result, a new feature, namely activation of monogamy of Bell inequality by considering local contextuality term is proposed. Finally, we uncovered an interesting new characteristic of Bell inequalities that trivial Bell expression can have no-signaling violations upon the addition of a third party with a single input. This implies that the study of Bell inequalities (and the corresponding classical polytopes) with single inputs for some of the parties becomes interesting.

Several open questions still remain. An interesting question is to settle the computational complexity of the chordal decomposition method for identifying monogamies. Another important question is to show  that local state-independent contextuality inequalities cannot replace the state-dependent ones in the monogamies we have formulated in the paper. If true, this would lend added strength to the resource-theoretic character of the trade-offs as suggested in \cite{CKK14}. It is also relevant to study the correspondence between such monogamy relation and the feature in which nonlocality can be revealed from local contextuality \cite{Cabello10,SCCP,NLCexpt}. It would also be of interest to identify the quantum boundaries of the novel trade-offs between non-locality and contextuality identified here, such as done for the case of Bell inequalities in \cite{TV06, KPR+11}. The extension of the above results to the multi-party scenario and network configurations such as in \cite{KPR+11} is of special interest, particularly with regard to cryptographic applications such as secret sharing \cite{SG01}. Finally, it would also be interesting to investigate the class of minimal non-contextuality inequalities whose addition activates the monogamy of any Bell inequality. 

{\em Acknowledgements.} We thank Remigiusz Augusiak and Marcin Paw\l{}owski for helpful discussions. This work is supported by the NCN grant 2014/14/E/ST2/00020 and the ERC grant QOLAPS.
%the EU grant RAQUEL (No. 323970). 

%\begin{thebibliography}{99}

\textbf{Supplemental Material.}

\textit{The graph-theoretic method to derive monogamy relations.---}

Here, we provide a brief explanation of the graph-theoretic method used in the main text, for more details see \cite{BP93, HL09}. We consider finite, simple and undirected graphs $G = (V(G), E(G))$. A subset $S \subset V(G)$ is said to be a separator of $G$ if two vertices in the same connected component of $G$ are in two distinct connected components of the graph $G \setminus S$. The set $S$ is said to be a $(u,v)$-\textit{separator} if it separates the vertices $u$ and $v$. A maximal clique tree of $G$ is a tree graph $T$ satisfying the following three conditions: (i) vertices of $T$ are associated with the maximal cliques of $G$, (ii) the edges of $T$ correspond to minimal separators, (iii) for any vertex $v \in V(G)$, the cliques containing $v$ yield a subtree of $T$. The reduced clique graph $\mathcal{C}_r(G)$ of a chordal graph introduced in \cite{HL09} is a graph whose vertices are the maximal cliques of $G$ and two cliques are joined by an edge if and only if their intersection separates them, i.e., for every $v_1 \in C_1 \setminus (C_1 \cap C_2)$ and $v_2 \in C_2 \setminus (C_1 \cap C_2)$ the vertex set $(C_1 \cap C_2)$ separates them. It is known \cite{HL09} that chordal graphs can equivalently be characterized as those graphs that admit a maximal clique tree. 

%A vertex $v \in V(G)$ is said to be simplicial if its neighborhood is a clique, where neighborhood refers to the vertices adjacent to $v$ in $G$. For $|V(G)| = n$, an ordering $v_1, \dots, v_n$ of the vertices if for every $k \in \{1, \dots, n-1\}$, $v_k$ is a simplicial vertex in the induced graph on the vertices $v_{k+1}, \dots, v_n$. 
%The crucial statement that allows the chordal graph to admit a joint probability distribution is the following lemma from \cite{}

\begin{prop}
	\label{prop-chordal}[see \cite{RSK}]
	A set of observables $A_1, A_2, \dots, A_n$ admits a joint probability distribution if the commutation graph $G(A_1, \dots, A_n)$ is chordal. 
\end{prop}
\begin{proof} 
	The proof is by construction of a joint probability distribution that reproduces all the marginal distributions that can be observed in the experiment. We first recall the construction of a maximal clique tree for a chordal graph (see \cite{BP93,GHP95} for an explanation of this concept). For a chordal graph $G$, the reduced clique graph $\mathcal{H}_c(G)$ introduced in \cite{HL09} is a graph whose vertices are the maximal cliques of $G$, and two maximal cliques are joined by an edge iff their intersection is a separator, where the set of common vertices $C_1 \cap C_2$ to two maximal cliques $C_1, C_2$ is a separator if upon its deletion from the graph, the remaining vertices in the two cliques become disconnected. Any edge in $\mathcal{H}_c(G)$ connecting  $C_1$ and $C_2$ is labeled with $C_1 \cap C_2$. The following Lemma from \cite{HL09} is crucial to the construction of the joint probability distribution. \\
	%Let us recall the notion of a maximal clique tree which will be crucial to the construction of the joint probability distribution. \\
	
	{\bf Lemma} \textit{[\cite{HL09}]
		Let $G$ be a chordal graph. 
		Consider any triangle in $\mathcal{H}_c(G)$ connecting maximal cliques $C_1, C_2$ and $C_3$ of $G$. Then two of the edge labels $C_1 \cap C_2$, $C_2 \cap C_3$ and $C_2 \cap C_3$ are equal and included in the third. Maximal clique trees $CT(G)$ of $G$ can be constructed by deleting edges from the triangles of $\mathcal{H}_c(G)$ that share a common edge label with another edge, all maximal clique trees define the same set of separators $\{C_i \cap C_j \}$.  } 
	%The maximal clique trees of $G$ correspond to the maximum spanning trees of $\mathcal{H}_c(G)$ when the edges are labelled with the size of the corresponding minimal separator, all maximal clique trees of $G$ correspond to the same set of minimal separators. 
	%
	% with the three minimal separators labeling its edges. Then two of these minimal separators must be equal and included in the third. The maximal clique trees correspond to maximum spanning trees of $\mathcal{C}(G)$.
	
	Thus, for a commutation graph $G$ representing a set of observables $\{A_1, \dots, A_n\}$, each vertex $v$ in the maximum clique tree $CT(G)$ is labelled by a clique $C_v$ in $G$ and each edge $e$ in $CT(G)$ is labelled by a minimal separator $C_u \cap C_v$. The vertex and edge labels of the maximum clique tree then give the construction of the joint probability distribution as 
	\begin{equation}
	P(A_1, \dots, A_n) := \frac{\prod_{v \in V(CT(G))} P(C_v)}{\prod_{(u,v) \in E(CT(G))} P(C_u \cap C_v)}.
	\end{equation}
	The construction of the spanning clique tree guarantees that the above distribution reproduces every measurable distribution (i.e., the distribution of every clique in $G$) as it's marginal.
\end{proof}

\begin{figure}[!htb]
	\centering
	\includegraphics[scale=.59]{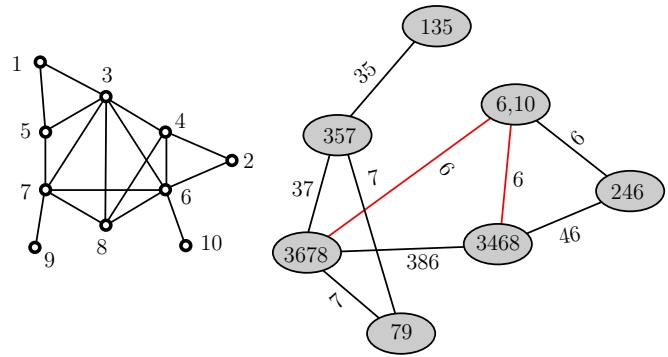}
	\caption{A chordal graph (left) and its reduced clique graph (right). Removing the red edges produces a maximal clique tree of the graph. A joint probability distribution for the observables can be constructed as a product of the vertex labels over the product of the edge labels, i.e., $\frac{\tiny{P(1,3,5)P(3,5,7)P(3,6,7,8)P(3,4,6,8)P(2,4,6)P(7,9)P(6,10)}}{\tiny{P(3,5)P(3,7)P(3,6,8)P(4,6)}P(7)P(6)}$.}
	\label{fig:chordal-gr}
\end{figure}
An example of a chordal graph and its maximal clique tree is presented in Figure \ref{fig:chordal-gr}. \\

\textit{Is a decomposition into chordal graphs necessary to derive monogamy?--- }

In this section, we study the question whether all no-signaling monogamies arise as a result of decomposition into chordal graphs. We exhibit a counter-example of a no-signaling monogamy relation that we prove cannot be derived in this manner, showing that the graph-theoretic decomposition while sufficient is not necessary. 
\begin{prop}
	There exist tight no-signaling monogamy relations of Bell inequalities (in two parties with three $\pm1$-valued settings each scenario) that do not arise from chordal decompositions of the commutation graph representing the corresponding observables. 
\end{prop}  
\begin{proof}
	The counter-example is given by the correlation Bell inequality (XOR game) represented in Fig \ref{xor}, which is a generalization of the CHSH expression to three inputs. The inequality is explicitly given as
	\begin{eqnarray}\label{xorineq}
	\mathcal{I}_{AB} &=& \langle A_1 B_1 \rangle - \langle A_1 B_2 \rangle - \langle A_1 B_3 \rangle + \langle A_2 B_1 \rangle +\langle A_2 B_2 \rangle  \nonumber \\ && - \langle A_2 B_3 \rangle + \langle A_3 B_1 \rangle + \langle A_3 B_2 \rangle + \langle A_3 B_3 \rangle \leq 5. \nonumber \\
	\end{eqnarray}
	The no-signaling value of this inequality is $9$ and a monogamy relation holds when the Bell experiment is performed by Alice together with Bob and Charlie, i.e. in any no-signaling theory 
	\begin{equation}\label{eq:counter-ex}
	\mathcal{I}_{AB} + \mathcal{I}_{AC} \leq 10.
	\end{equation} 
	That such a relation holds can be directly verified by means of a linear program. 
	
	We now show that no chordal decomposition of the commutation graph of the observable set $\{A_1,\dots,C_3\}$ exists to recover the relation (\ref{eq:counter-ex}). A total of $18$ winning constraints (correlation and anti-correlation edges in the graph) enter in the relation (\ref{eq:counter-ex}). In order to derive the monogamy relation, out of those 8 cannot be satisfied. 
	Firstly, since each of the decomposed induced subgraphs is chordal, it cannot contain more than one observable from $\{B_1,B_2,B_3\}$ and  $\{C_1,C_2,C_3\}$. So, it may have maximum six edges in which maximum 2 winning constraints can be dissatisfied. On the other hand, to discard 2 winning constraints the decomposed graph should have a four cycle. Thus the only way to distribute the 18 edges, such that there are 8 winning constraints are dissatisfied, is to consider 4 chordal four-cycles and two additional edges. To have 4 chordal four-cycles, one of $\{B_1,B_2,B_3,C_1,C_2,C_3\}$ must appear twice in those cycles. Hence one of $\{B_1,B_2,B_3,C_1,C_2,C_3\}$ must be a part of four different edges, which is a contradiction since all of these vertices appear in three different edges in the commutation graph. \\
	Alternatively, one can consider all the induced chordal decomposition which is a partition of these $18$ edges.  Out of these, it can be readily seen that any partition which contains fewer than $4$ edges can be discarded. This leaves a total of $3$ possible ways of chordal decompositions which can be exhaustively checked.
\end{proof}

\begin{figure}[!htb]
	\centering
	\includegraphics[scale=.6]{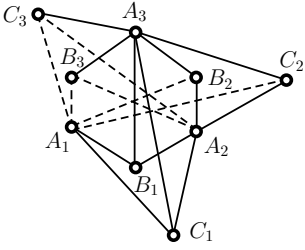}
	\caption{The commutation graph $G[\mathcal{I}_{AB}]+G[\mathcal{I}_{AC}]$ given in \eqref{xorineq} that serves as a counter-example to the necessity of decompositions into chordal graphs for the existence of monogamies. The \textit{dashed} line represents the contradiction edge.}
	\label{xor}
\end{figure}

\textit{Monogamy of cycle inequalities.---}

Without loss of generality, the cycle inequality can be written as,
\begin{equation}\begin{split}\label{chain}
\mathcal{I}(n) =    \sum\limits_{i=1}^{n-1} \langle A_iA_{i+1}\rangle -\langle A_nA_{1}\rangle 
\leq n-2.
\end{split}\end{equation}  Let's denote the corresponding commutation graph by $G(A_i,e_i)$ where the edge $e_i := (A_i,A_{i+1})$ and the sum of the index is taken to be modulo $n$.
\begin{prop}
	For any two inequalities $\mathcal{I}(n) \leq (n-2)$ and $\mathcal{I}(m) \leq (m-2)$ corresponding to cyclic commutation graphs $G(A_i,e_i)$ and $G'(A'_i,e'_i)$ having  common $k(\geq2)$ vertices, the following monogamy relation holds
	\begin{equation}\label{mr}
	\mathcal{I}(n)+\mathcal{I}(m) \overset{ND}{\leq} m+n-4
	\end{equation} in any theory satisfying the no-disturbance principle. For the monogamy to hold in the nonlocal-contextual and contextual-contextual scenario, suitable additional commutation relations are required.
\end{prop}
\begin{proof} 
	Without loss of generality, we assume $n\leq m$ and the edge $e_n = (A_n,A_1)$ is the contradiction edge of $G$. Let's $A_{i_1}(A'_{j_1}),A_{i_2}(A'_{j_2}),..., A_{i_k}(A'_{j_k}) $ are the common vertices. The requirement for deriving monogamy relation is to decompose the whole graph into two induced chordal subgraphs such that each possesses one contradiction edge in a cycle of more than three edges. Here we provide such constructions considering two different situations separately.\\
	$(i)$ When the contradiction of $G'$ belongs to $ \{e'_{j_k},e'_{j_k+1},...,e'_{j_1-1}\}$, decompose the whole graph into two induced cycles involving $(A_{i_k},A_{i_k+1},...,A_{i_1},A'_{j_1},...,A'_{j_{k}})$ and $(A'_{j_k},A'_{j_k+1},...,A'_{j_1},A_{i_1},...,A_{i_{k}})$.  \\
	$(ii)$ When the contradiction of $G'$ belongs to $\{e'_{j_l},e'_{j_l+1},...,e'_{j_{l+1}-1}\}$ for some $l \neq k$, consider two induced subgraphs involving $(A_{i_l},A_{i_l+1},...,A_{i_{l+1}},A'_{j_1},...,A'_{j_k})$ and 
	$(A_{i_1},...,A_{i_l},A_{i_{l+1}},...,A_{i_1-1},A'_{j_k},...,A'_{j_1})$.
	\\
	If these two subgraphs are chordal, then this decomposition satisfies all the requirements for the monogamy relations to hold. In nonlocal-nonlocal scenarios this is true by construction. In nonlocal-contextual and contextual-contextual scenarios, this extra condition is required.
\end{proof}

\begin{figure}[!htb]
	\centering
	\includegraphics[scale=.45]{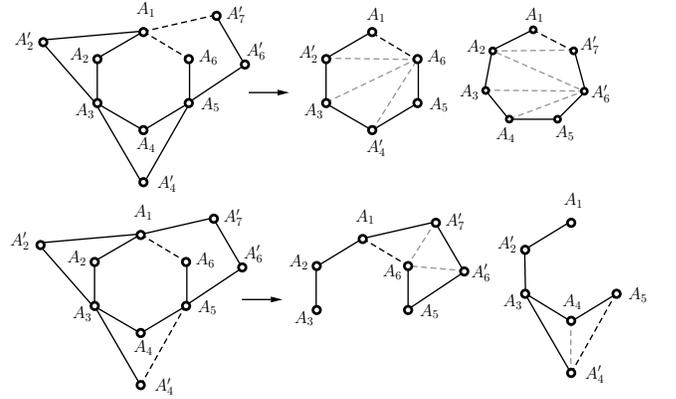}
	\caption{An example showing decomposition into induced chordal subgraphs of the commutation graph corresponding to two cyclic inequalities given in \eqref{chain} taking $n=6,m=7$. Here, the \textit{dashed} line represents the contradiction edge and \textit{grey} lines denote the additional commutation relations. The upper and lower illustration correspond to situation $(i)$ and $(ii)$ respectively as described in Proposition 3. }
	\label{cycle}
\end{figure}

The chordal decomposition presented in Proposition 3 is also applicable to the multiple outcome scenarios by considering directed commutation graphs. The cycle Bell-inequality has been generalized \cite{BKP} to the $d$-outcome scenario and its monogamous property has been noted there. This inequality can also be generalized to the contextuality scenario as,
\begin{equation}\label{BKP}
\mathcal{I}(n) = \sum_{i=1}^{n-1} \langle [A_i - A_{i+1}] \rangle +  \langle [A_{n} - A_{1}-1] \rangle \geq d-1
\end{equation} where all the observables $A_i$ has $d \in \{0,1,...,d-1\}$ possible outcomes and  $[X]$ denotes $X$ modulo $d$. Let's represent this inequality by directed cyclic commutation graph. In this case, the vertices will have $d\in \{0,1,...,d-1\}$ possible values. The edges $e_i = ( A_{i},A_{i+1})$ with the direction $i\rightarrow i+1$ represents the observed value of the quantity $[A_i - A_{i+1}]$ whereas the edge representing $[A_{n} - A_{1}-1]$ is defined as the contradiction edge. The monogamy relation $\mathcal{I}(n)+\mathcal{I}(m) \overset{ND}{\geq} 2(d-1)$ can be obtained by the method of Proposition 3 only in the situation $(i)$.\\

%\textit{Monogamy activation of $\mathcal{I}_{3322}$.---}
%
%We consider the $I_{3322}$ and $\mathcal{I}(5)$ inequalities as in the text.
%\begin{equation}\label{i3322}
%\begin{split}
%\mathcal{I}^{A_1A_4A_6}_{B_1B_2B_3} = & \langle A_1 \rangle + \langle A_4 \rangle +\langle B_1 \rangle + \langle B_2 \rangle - \langle A_1B_1 \rangle  - \langle A_1B_2 \rangle - \langle A_1B_3 \rangle \\
%& - \langle A_4B_1 \rangle - \langle A_4B_2 \rangle + \langle A_4B_3 \rangle - \langle A_6B_1 \rangle + \langle A_6B_2 \rangle \leq 4.
%\end{split}\end{equation} 
%\begin{equation}\begin{split}\label{chain5}
%\mathcal{I}(5) =    \langle A_1A_{2}\rangle +  \langle A_2A_{3}\rangle +  \langle A_3A_{4}\rangle + \langle A_4A_{5}\rangle -\langle A_5A_{1}\rangle 
%\leq 3.
%\end{split}\end{equation}
%
%A single $\mathcal{I}_{3322}$ inequality does not exhibit a monogamy with the local non-contextuality inequality $\mathcal{I}(5)$.  It can be readily verified that the given box satisfies all the no-signaling constraints and violates both the inequality, $\mathcal{I}_{3322}=4.333, \mathcal{I}(5)=4$. 
% Also, as found in \cite{CG04}, two $\mathcal{I}_{3322}$ inequalities also do not exhibit any monogamy between them. As explained in the text, on the other hand, the minimal local contextuality term $\mathcal{I}(5)$ activates a monogamy between two $\mathcal{I}_{3322}$ inequalities.  \\

\begin{table}[h!]
	\begin{center}
		\caption{In the following table the probability distribution is given which satisfies all the no-signaling (no-disturbance) constraints and violates both the inequality, $\mathcal{I}_{3322}=4.333, \mathcal{I}(5)=4$. } 
		\begin{tabular}{| l |  m{0.75cm}  m{0.75cm} | m{0.75cm}  m{0.75cm} | m{0.75cm}  m{0.75cm} | }
			\hline
			&  \multicolumn{2}{|c|}{$B_1$} &\multicolumn{2}{c}{$B_2$} &\multicolumn{2}{|c|}{$B_3$} \\
			\hline
			%  \tiny{Alice/Bob} & $B_1$  & $B_2$&  $B_3$ & \\ \hline \hline
			\multirow{4}{4em}{$A_1 A_2$}       &  \scriptsize{1/4}   & \scriptsize{1/2} & \scriptsize{1/4}  & \scriptsize{1/2}& \scriptsize{1/4}  & \scriptsize{1/2} \\ 
			& $\cdot$ & $\cdot$ &$\cdot$ &$\cdot$  & $\cdot$& $\cdot$ \\ 
			&    $\cdot$   &$\cdot$ &$\cdot$ &$\cdot$ &$\cdot$ &  $\cdot$\\ 
			&    \scriptsize{1/4}& $\cdot$ & \scriptsize{1/4} & $\cdot$& \scriptsize{1/4} & $\cdot$ \\ \hline
			\multirow{4}{4em}{$A_2 A_3$}  		 &    \scriptsize{1/4} & \scriptsize{1/2}  & \scriptsize{1/4}  & \scriptsize{1/2} & \scriptsize{1/4} & \scriptsize{1/2} \\ 
			& $\cdot$  & $\cdot$ & $\cdot$& $\cdot$ &$\cdot$ & $\cdot$ \\ 
			& $\cdot$ & $\cdot$& $\cdot$& $\cdot$& $\cdot$& $\cdot$ \\ 
			&\scriptsize{1/4}& $\cdot$& \scriptsize{1/4} & $\cdot$& \scriptsize{1/4} & $\cdot$  \\ \hline
			\multirow{4}{4em}{$A_3 A_4$}     & \scriptsize{1/4} &  \scriptsize{1/3} &  \scriptsize{1/4} &  \scriptsize{1/3} &  \scriptsize{1/4} &  \scriptsize{1/3}\\ 
			& $\cdot$  & \scriptsize{1/6}  & $\cdot$& \scriptsize{1/6}  & $\cdot$  &  \scriptsize{1/6} \\ 
			& $\cdot$  & $\cdot$&$\cdot$ & $\cdot$& $\cdot$&$\cdot$  \\ 
			&  \scriptsize{1/4} &  $\cdot$ &  \scriptsize{1/4} &  $\cdot$ &  \scriptsize{1/4} &  $\cdot$ \\ \hline
			\multirow{4}{4em}{$A_4 A_5$}          &    \scriptsize{1/4} &  \scriptsize{1/6} &  \scriptsize{1/4} &  \scriptsize{1/6} &  \scriptsize{1/4} &  \scriptsize{1/6}\\ 
			& $\cdot$  & \scriptsize{1/6}  & $\cdot$ & \scriptsize{1/6}  & $\cdot$ & \scriptsize{1/6}  \\ 
			& $\cdot$   &$\cdot$ & $\cdot$& $\cdot$&$\cdot$ &  $\cdot$\\ 
			&   \scriptsize{1/4} &  \scriptsize{1/6} &  \scriptsize{1/4} &  \scriptsize{1/6} &  \scriptsize{1/4} &  \scriptsize{1/6}\\ \hline
			\multirow{4}{4em}{$A_5 A_1$} 		&	$\cdot$   & \scriptsize{1/6}  & $\cdot$& \scriptsize{1/6}  & $\cdot$ &  \scriptsize{1/6} \\            
			&   \scriptsize{1/4} &  $\cdot$ &  \scriptsize{1/4} &  $\cdot$ & \scriptsize{1/4} &  $\cdot$ \\ 
			& \scriptsize{1/4} &  \scriptsize{1/3} &  \scriptsize{1/4} &  \scriptsize{1/3} &  \scriptsize{1/4} &  \scriptsize{1/3}\\
			&  $\cdot$ & $\cdot$& $\cdot$&$\cdot$ & $\cdot$& $\cdot$ \\  \hline
			\multirow{2}{4em}{$A_6$}          & $\cdot$ &  \scriptsize{1/2} &  \scriptsize{1/2} & $\cdot$ &  \scriptsize{1/4} &  \scriptsize{1/4} \\
			&   \scriptsize{1/2} & $\cdot$ & $\cdot$ &  \scriptsize{1/2}&  \scriptsize{1/4} &  \scriptsize{1/4} \\\hline
			%  $A_1 A_4$ & 01 & & & & & &  \\ \hline
		\end{tabular}
	\end{center}
\end{table}

%\begin{thebibliography}{99}

\end{document}